\documentclass[twocolumn]{autart}    
\makeatletter
\newif\if@restonecol
\makeatother

\usepackage[linesnumbered,ruled,lined,commentsnumbered,norelsize]{algorithm2e}
\SetAlFnt{\scriptsize}
\SetAlCapNameFnt{\scriptsize}
\usepackage{algpseudocode}
\usepackage{graphicx}          

\usepackage[]{todonotes}
\usepackage[T1]{fontenc}
\let\theoremstyle\relax
\usepackage{amsfonts}
\usepackage{amsthm, amscd}
\usepackage{amsmath}
\usepackage{harvard} 
\usepackage{amssymb}
\usepackage{picins}
\usepackage[nodisplayskipstretch]{setspace}
\usepackage{epstopdf}
\usepackage{url}
\usepackage{parskip}
\usepackage{lmodern}
\usepackage{xcolor}
\usepackage{ifthen}
\usepackage{float}
\usepackage[tight,footnotesize]{subfigure}
\usepackage[caption=false,font=footnotesize]{subfig}
\usepackage{booktabs}  %
\usepackage{diagbox}
\usepackage{caption}
\usepackage{chngcntr}
\let\emptyset\varnothing
\usepackage{enumitem}
\usepackage{mathtools}
\usepackage{nccmath}
\usepackage{etoolbox}

\renewenvironment{thebibliography}[1]{
  \begin{oldthebibliography}{#1}
    \setlength{\itemsep}{2pt}
    \setlength{\parskip}{0pt plus 0.2ex}
}
{
  \end{oldthebibliography}
}

\newtheorem{theorem}{Theorem}
\newtheorem{assumption}{Assumption}
\newtheorem{lemma}{Lemma}

\newtheorem{remark}{Remark}
\newtheorem{definition}{Definition}

\def\cA{{\mathcal A}}

\def\cF{{\mathcal F}}
\def\cC{{\mathcal C}}
\def\cE{{\mathcal E}}
\def\cU{{\mathcal U}}
\def\cR{{\mathcal R}}
\def\cI{{\mathcal I}}

\def\cN{{\mathcal N}}
\def\cL{{\mathcal L}}

\newcommand{\pump}[2]{
	\ifthenelse{\equal{#2}{}}{g}{
		\ifthenelse{\equal{#1}{0}}{#2}{
			\ifthenelse{\equal{#1}{1}}{g(#2)}{
				g^{#1}(#2)
			}
		}
	}
}
\newcommand{\pumpdir}[2]{
	\ifthenelse{\equal{#1}{0}}{#2}{
		\ifthenelse{\equal{#1}{1}}{K #2 + 1}{
			\ifthenelse{\equal{#1}{2}}{K^2 #2 + K + 1}{
				\ifthenelse{\equal{#1}{3}}{K^3 #2 + K^2 + K + 1}{
					K^{#1} #2 + \sum_{i=0}^{i<#1} K^i
				}
			}
		}
	}
}
\newcommand{\pumpinv}[2]{
	\ifthenelse{\equal{#1}{0}}{#2}{
		\ifthenelse{\equal{#1}{1}}{\left\lceil \frac{#2 - 1}{K} \right\rceil}{
			\ifthenelse{\equal{#1}{2}}{\left\lceil \frac{#2 - K - 1}{K^2} \right\rceil}{
				\ifthenelse{\equal{#1}{3}}{\left\lceil \frac{#2 - K^2 - K- 1}{K^3} \right\rceil}{
					\left\lceil \frac{#2 - \sum_{i=0}^{i<#1} K^i}{K^{#1}} \right\rceil
				}
			}
		}
	}
}

\begin{document}
\newcommand{\bff}{{\boldsymbol{f}}}
\newcommand{\bg}{{\boldsymbol{g}}}
\newcommand{\bI}{{\boldsymbol{I}}}
\newcommand{\bu}{{\boldsymbol{u}}}
\newcommand{\bD}{{\boldsymbol{D}}}
\newcommand{\bv}{{\boldsymbol{v}}}
\newcommand{\bS}{{\boldsymbol{S}}}
\newcommand{\bCO}{{\boldsymbol{CO}}}
\newcommand{\bc}{{\boldsymbol{c}}}	
\newcommand{\bF}{{\boldsymbol{F}}}	
\newcommand{\bm}{{\boldsymbol{m}}}	
\newcommand{\RomanNumeralCaps}[1]
{\MakeUppercase{\romannumeral #1}}
\newcommand{\rom}[1]{\uppercase\expandafter{\romannumeral #1\relax}}
\def\QED{~\rule[-1pt]{5pt}{5pt}\par\medskip}
\def\eps{\epsilon}
\def\xh{\hat{x}}
\def\yh{\hat{y}}
\def\xd{\dot{x}}
\def\yd{\dot{y}}
\def\Vd{\dot{V}}
\def\Vd{\dot{V}}
\def\ed{\dot{e}}
\def\xu{{\bf x}}
\def\xhu{{\hat{\bf x}}}
\def\qh{\hat{q}}
\def\xt{\tilde{x}}
\def\yb{\bar{y}}
\def\cA{{\cal A}}
\def\cL{{\cal L}}
\def\xb{\bar{x}}
\def\cR{{\cal R}}
\def\cC{{\cal C}}
\def\cS{{\mathcal S}}
\def\cE{{\cal E}}
\def\cK{{\mathcal K}}
\def\cU{{\cal U}}
\def\cI{{\cal I}}
\def\cP{{\mathcal P}}
\def\cQ{{\cal Q}}
\def\cW{{\cal W}}
\def\cU{{\cal U}}
\def\cCh{\hat{{\cal C}}}
\def\bR{\mathbb{R}}
\def\bZ{\mathbb{Z}}
\def\cC{{\cal C}}
\def\cS{{\mathcal S}}
\def\cD{{\mathcal D}}
\def\cA{{\cal A}}
\def\cU{{\cal U}}
\begin{frontmatter}

\title{Fully distributed singularity-free prescribed-time stabilization of the continuous-time generalized adaptive Bellman-Ford algorithm}

\author[First]{Yuanqiu Mo}\ead{yuanqiumo@seu.edu.cn},
\author[Second]{Jian Qin}\ead{15240231458@163.com},
\author[Third]{Soura Dasgupta}\ead{soura-dasgupta@uiowa.edu}

\address[First]{Department of System Science, School of Mathematics, Southeast University, Nanjing 211189, China}
\address[Second]{School of Mathematics, Sichuan University, Chengdu 610065, China}
\address[Third]{Department of Electrical and Computer Engineering, University of Iowa, Iowa City (IA) 52242 USA}

\begin{abstract}                          
Building upon the well-established distributed biased min-consensus protocol, which serves as an efficient approach to address the shortest path problem in a distributed fashion, the continuous-time generalized adaptive Bellman-Ford algorithm (GABF) introduces flexibility by accommodating various forms of distance metrics. This adaptability makes GABF suitable for more complex scenarios, such as time-dependent shortest path problem and robotic path planning. However, existing research on this protocol primarily focuses on asymptotic stability, providing no insights into convergence speed, which limits its practical applications. To address this gap, this paper proposes two control strategies that achieve prescribed-time stabilization of GABF by ensuring its convergence to the stationary value within a user-defined time, thereby broadening its applicability. Simulation scenarios, including robotic manipulator path planning with real-world data and learning-based path planning, are provided to validate the effectiveness of the proposed approaches.
\end{abstract}

\end{frontmatter}

\section{Introduction}
In contrast to classical consensus protocols that  focus on mitigating  differences between  states of individual agents \cite{Qinconsensus}, the distributed biased min-consensus protocol (DBMC) \cite{Zhang2017} enables advanced computations by solving the shortest path problem \cite{bellman1958routing} in a distributed manner. As a result, DBMC finds its merits in various path-finding related applications \cite{ran2024fast,yao2020hierarchical}. The stability results for DBMC include global asymptotic stability \cite{Zhang2017}, regional exponential stability, and convergence within a user-defined time under stringent conditions on the initial states and graph-dependent control parameters; see \cite{mo2025} and references therein. Further, the discretized version of DBMC has been shown to achieve finite time convergence in several studies \cite{zhang2017perturbing,yao2020hierarchical,ran2024fast}, with convergence times depending on the initial states and graphical parameters. 

However, DBMC can only find a path between two nodes in a graph such that the sum of its constituent edge weights is minimized. This approach may not be suitable for more complex scenarios. For instance, in robotic path planning, in addition to the distance cost, one must  also  consider the cost of manipulation which  is position-dependent \cite{shen2023adaptive}. 
 Towards this end, the continuous-time generalized adaptive Bellman-Ford algorithm (GABF) is proposed \cite{mo2021global}, as the continuous-time counterpart to the discrete-time $G$-block described in \cite{mo2022stability}. The strength of GABF lies in its generality and flexibility: instead of prescribing a specific form of dynamics, it defines a set of properties that the evolution must satisfy. This abstraction broadens its applicability across diverse scenarios. In contrast to DBMC, research on design and stability of GABF remains limited. \textcolor{red}{Although self-stabilization and finite time convergence of the discrete-time version of GABF have been established in \cite{VABDP-ACM} and \cite{mo2022stability,miller2025convergence}, the former provides no explicit bound on the convergence time, while the latter yields convergence times that depend on both the graph diameter and the initial conditions. Moreover, existing works on GABF itself are limited to asymptotic stability results \cite{mo2021global}. In contrast, prescribed-time control \cite{PTISSAuto,krishnamurthy2020dynamic} provides a user-defined convergence time guarantee in the continuous-time setting. From a practical perspective, although digital implementations inevitably introduce numerical errors and cannot achieve exact convergence, prescribed-time control still provides a guarantee on the time by which the system reaches a prescribed accuracy. This feature is fundamentally different from asymptotic or finite time convergence, whose convergence times are either infinite or dependent on initial conditions and system parameters, and therefore cannot be certified \emph{a priori}. Such a guarantee is particularly valuable in time-critical applications, such as real-time path planning \cite{shen2023adaptive,lavalle2006planning}, where response times must be certified before deployment. Therefore, developing prescribed-time control strategies for GABF is of both theoretical and practical importance.}

This paper presents two control strategies that achieve prescribed-time stabilization of GABF. Although prescribed-time consensus and distributed optimization protocols \cite{wang2018leader,zhou2024fully,zheng2023specified,WEI,Prescribed1,Prescribed2} admit Lyapunov-based analysis, extending such techniques to GABF is nontrivial because the consensus structure is replaced by a nonlinear, non-differentiable Bellman operator. Instead, the prescribed-time stability of the proposed methods is analyzed by exploiting structural properties of the shortest path problem. 
By characterizing each node’s state trajectory using that of its parent node, it is shown that the nominal GABF is regionally exponentially stable. Building on this result, the first strategy employs a time-space deformation \cite{orlov2022time} to establish that the prescribed-time controlled nominal GABF admits a continuously differentiable solution that converges exactly at the prescribed time.
The second strategy overcomes the limitations of the first, namely the infinite control gain at the prescribed time and the initial overestimation requirement. It first embeds GABF into a fixed time control framework, where node states enter the framework in ascending order of their shortest path lengths, thereby establishing global finite time convergence, and then incorporates a prescribed-time control law to achieve fully distributed prescribed-time stabilization before the preset time. The contribution of this work can be summarized as follows:
\begin{itemize}[leftmargin=*]
	\item This paper studies the prescribed-time stabilization of GABF that can tackle various shortest path-finding problems. Unlike existing results on asymptotic stability \cite{mo2021global} or finite time convergence dependent on graph structure and initial conditions \cite{VABDP-ACM,mo2022stability},
	the developed prescribed-time stability enables GABF to stabilize within a user-defined time. This enhancement extends its applicability to more demanding scenarios, such as robotic path planning \cite{yao2020hierarchical,shen2023adaptive}, where timely resolution of path-finding tasks is critical. 
	\item To eliminate the infinite gain inherent in prescribed-time controllers \cite{song2023prescribed,zheng2023specified,wang2018leader,krishnamurthy2020dynamic} and relax the initial overestimation assumption in DBMC \cite{mo2025}, this paper integrates a fixed time control framework with the prescribed-time control strategy. Leveraging the global finite time convergence of the former, the resulting scheme achieves fully distributed, singularity-free prescribed-time stabilization of GABF. Notably, unlike related works \cite{GUOGE,WEI}, which require undirected graphs, the proposed non-Lyapunov-based analysis accommodates directed graphs. 
\end{itemize}
The remainder of the paper is organized as follows. Section \ref{sec:pre} introduces preliminaries. Section \ref{sec:ptandppt} presents the prescribed-time control strategies. Section \ref{sec:simulations} provides simulation results, and Section \ref{sec:con} concludes the paper.

\section{Preliminaries}\label{sec:pre}
This paper considers a directed graph $G=(V,E)$, where $V = \{1,\cdots,n\}$ denotes the set of nodes and $E$ the set of directed edges. Each edge $(i, j) \in E$ indicates that information can flow from node $j$ to node $i$. The set of out-neighbors of node $i \in V$ is denoted by $\cN(i) = \{ j \in V ~|~ (i,j) \in E  \}$, and $i \notin \cN(i)$ for all $i \in V$. The weight of the edge $(i,j)$ is denoted by $w_{ij}$. A particular subset of nodes forms the set of source nodes, and is denoted by $S \subsetneq V$. Furthermore, define $\mathbb{R}$ as the set of real numbers. For $\Pi \subseteq \mathbb{R}$, define $\Pi_{\geq c} = \{x \in \Pi~|~ x \geq c\}$. 

Let $t_0 = 0$ and $x_i(t)$ denote the initial time and state of node $i$. The nominal GABF specifies the recursion:
\begin{equation}\label{eq:protocol}
	\xd_i(t) \!=\! 
	\begin{cases}
		0, \!\!\!\!& i \in S\\
		-\eta\left(x_i(t) - \min_{j \in \cN(i)}\left\{f_i(x_j(t),w_{ij})\right\}\right), \!\!\!\!& i \notin S
	\end{cases},
\end{equation}
where $f_i: \mathbb{R}_1 \times \mathbb{R}_2 \rightarrow \mathbb{R}_1$ with $\mathbb{R}_1, \mathbb{R}_2 \subseteq \mathbb{R}$ is the kernel function of node $i$, and $\eta > 0$.
In (\ref{eq:protocol}), the states of source nodes remain anchored at their initial values, while the state of each non-source node evolves according to a kernel function that depends on the states of its neighboring nodes and the weights of the connecting edges. In particular, the kernel function $f_i$ has the following two properties: 1) $f_i$ is progressive, i.e., with $\zeta > 0$,
	\begin{equation}\label{eq:progressive}
		f_i(a, b) \geq a + \zeta, ~ \forall a \in \mathbb{R}_1, b \in \mathbb{R}_2;
	\end{equation}
2) $f_i$ is strictly increasing in its first argument.

In contrast to  DBMC using $f_i(x_j(t), w_{ij}) = x_j(t) + w_{ij}$ to minimize the sum of edge weights along the path, the general nature of $f_i$ allows for diverse representations of edge weights. For instance, the manipulability measure, a crucial factor in the cost function of robotic path planning, depends on the robot's position \cite{shen2023adaptive}. 
	
Define $x^*_i$ as the optimal cost required to travel from node $i$ to its nearest source using an optimal policy based on the kernel function specified in (\ref{eq:protocol}). Employing the Bellman's principle of optimality \cite{bellman1958routing}, $x^*_i$ satisfies the following nonlinear system of equations:
\begin{equation}\label{eq:bellman}
	x^*_i = \begin{cases}
		x_i(0), &i \in S \\
		\min_{j \in \cN(i)} \{f_i(x^*_j,w_{ij}) \},& i \notin S
	\end{cases}.
\end{equation}
The basic principle underlying (\ref{eq:bellman}) is that any subpath of an optimal path from node $i$ to its nearest source is itself optimal. In other words, the optimal policy leading to $x^*_i$ also contains optimal policies from any intermediate node to its nearest source.


The following assumption holds throughout the paper.
\begin{assumption}\label{ass:main}
	The graph $G = (V, E)$ considered is directed, the source set $S \neq  \emptyset$ and $S \subset V$, each non-source node has at least one directed path to the source set,
	and $t_0 = 0$ is the initial time.	
\end{assumption}
Under Assumption \ref{ass:main}, GABF naturally fits the stochastic shortest path (SSP) framework of dynamic programming \cite{bertsekas1996neuro}.
\begin{remark}\label{re:SSP}
Let $U(i)$ be a finite set of admissible controls at node $i \in V$, and let $p_{ij}(u) \ge 0$ denote the transition probability from node $i$ to node $j$ under $u \in U(i)$, with $\sum_{j \in V} p_{ij}(u) = 1$. Let $g(i,u,j)$ denote the one-step cost incurred when transitioning from $i$ to $j$ under $u$. For $i \in S$, $p_{ii}(u) = 1$ and $g(i,u,i) = 0$ for all $u \in U(i)$. 
With any vector $J = [J(1), \cdots, J(n)]^\top$, 
the associated Bellman operator $\mathcal{T} : \mathbb{R}^{n} \to \mathbb{R}^{n}$ in SSP is defined as
	\begin{equation*}
		(\mathcal{T}J)(i) =
		\begin{cases}
			0, & i \in S, \\
			\displaystyle \min_{u \in U(i)} \sum_{j \in V} p_{ij}(u)\big( g(i,u,j) + J(j) \big), & i \notin S.
		\end{cases}
	\end{equation*}
In GABF, consider $\mathcal{T}$ with the following specifications:
	$U(i)=\mathcal{N}(i)$, $x_i(0)=0$ for $i\in S$, $p_{ij}(u)=1$ if $u=j$ and $0$ otherwise, and the cost
	$g(i,u,j)=f_i(J(j),w_{ij})-J(j)$.
	Then $\mathcal{T}$ reduces to
	$\mathcal{T}' : \mathbb{R}^{n} \to \mathbb{R}^{n}$, where
	\begin{equation}\label{eq:bellman-gabf}
		(\mathcal{T}'J)(i) =
		\begin{cases}
			x_i(0), & i \in S, \\[1mm]
			\displaystyle \min_{j \in \mathcal{N}(i)} f_i\big(J(j), w_{ij}\big),
			& i \notin S.
		\end{cases}
		\end{equation}
Moreover, the positivity of costs implied by \eqref{eq:progressive}, together with the graph connectivity in Assumption~\ref{ass:main}, ensures that assumptions in \cite[Proposition 2.1]{bertsekas1996neuro} are satisfied, and $\mathcal{T}'$ admits a unique optimal solution $J^\ast = \mathcal{T}'J^\ast$, which coincides with (3),  establishing the uniqueness of the optimal cost vector for GABF.
\end{remark}

We introduce the following definition to facilitate the subsequent analysis.
\begin{definition}\label{def:true}
	Given a connected and directed graph $G=(V,E)$, we call the minimizing $j$ on the right hand side of the second bullet of (\ref{eq:bellman}) a true parent node of $i$. As $i$ may have multiple true parent nodes, we use $\cP(i)$ to denote the set of true parent nodes of node $i$. A source node does not have any true parent node.
	
	Further, consider a sequence of nodes such that the predecessor of each node is one of its true parent nodes. We define $\cD(G)$, the effective diameter of $G$, as the longest length such a sequence can have in graph $G$.
\end{definition}
The directed graph in Figure \ref{fig:example1} consists of 9 nodes: 2 source nodes (nodes 1 and 9, shown in red) and 7 non-source nodes (shown in blue). The numbers above the edge represent the edge weights. $f_i: \mathbb{R}_{\geq 1} \times \mathbb{R} \rightarrow \mathbb{R}_{\geq 1}$ is set as $f_i(a,b) = 2a + |b|$ if $i$ is even and  $f_i(a,b) = a + 2|b|$ otherwise. Let $x_1^* = x_9^* = 1$. $x_4^*= 7$ is achievable via paths $4 \rightarrow 2 \rightarrow 1$ or $4 \rightarrow 3 \rightarrow 1$. Hence, $\cP(4) = \{2,3\}$ and $\cD(G) = 3$ indicated by paths like $6 \rightarrow 8 \rightarrow 9$. 

The following definition will be used to quantify the later stability analysis.
\begin{definition}\label{def:longest}
	We call a path from node $i$ to its nearest source $j \in S$ an optimal path, if it starts at $i$, ends with $j \in S$, and the predecessor of each node in the path is one of its true parent nodes. We call such an optimal path the longest optimal path if it has the most nodes among all optimal paths of $i$. The set $\cF_\ell$ is the set of nodes whose longest optimal paths to the source set have $\ell + 1$ nodes. In particular, $\cF_0 = S$, $\bigcup_{i \in \{0,\cdots, \cD(G) - 1  \}}\cF_i = V ~\mathrm{and}~ \cF_i \cap \cF_j = \emptyset, ~ \forall i \neq j$.
\end{definition}
\begin{figure}[h]
	\centering
	\includegraphics[width=0.9\columnwidth]{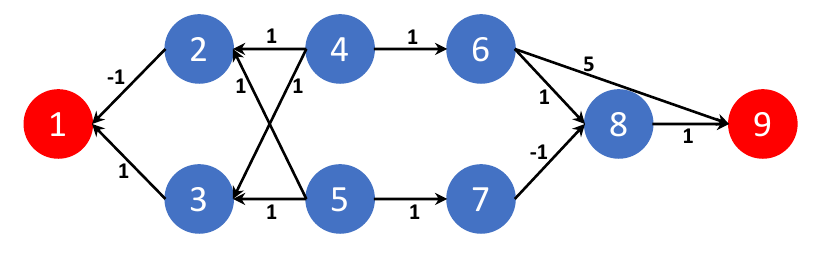}
	\caption{An example illustrating Definitions \ref{def:true} and \ref{def:longest}.} 
	\label{fig:example1}
\end{figure}
We again use Figure \ref{fig:example1} to illustrate Definition \ref{def:longest}. In this example, node 6 has two optimal paths to its nearest source (node 9): path $6 \rightarrow 8 \rightarrow 9$ and path $6 \rightarrow 9$, both of which yields $x_6^* = 7$.  By Definition \ref{def:longest}, $6 \in \cF_2$ but $6 \notin \cF_1$. 

\noindent
\begin{assumption}\label{ass:Lip}
	The kernel function $f_i$ is Lipschitz continuous in its first argument, i.e., there exists an $L > 0$ such that for all $ i \in V\setminus S$ and $b \in \mathbb{R}_2$
	\begin{equation}\label{eq:Lipschitz}
		|f_i(a_1,b) - f_i(a_2,b)| \leq L|a_1 - a_2|, ~ \forall a_1, a_2 \in \mathbb{R}_1.
	\end{equation}
\end{assumption}
The following overestimation assumption applies specifically to the first prescribed-time control strategy.
\begin{assumption}\label{ass:over}
All initial states are assumed to be strict overestimates, i.e., $x_i(0) > x^*_i$ for all $i \in V\setminus S$.
\end{assumption}

\begin{remark}\label{re:over}
Assumption \ref{ass:over} requires the use of certain global parameters, including the number of nodes and the maximum edge weight. Similar requirements also appear in many other prescribed-time control studies, e.g., \cite{wang2018leader,zheng2023specified,WEI}. In practice, such information can be estimated via distributed average and max consensus protocols, respectively \cite{averageconsensus,maxconsensus}. Alternatively, one may initialize the states to be sufficiently large, as Assumption \ref{ass:over} only requires them to be lower bounded.
	
\end{remark}

\section{Prescribed-time stabilization}\label{sec:ptandppt}
In this section, we introduce two control strategies to guide the trajectory of (\ref{eq:protocol}) to the stationary value (\ref{eq:bellman}) within a user-specified time. The underlying theoretical framework is summarized in Figure~\ref{fig:flowchart}. We first revisit the nominal GABF (\ref{eq:protocol}) and establish its regional exponential stability. Unless otherwise noted, all proofs are provided in the Appendix. 
\begin{figure}[h]
	\centering
	\includegraphics[width=1\columnwidth]{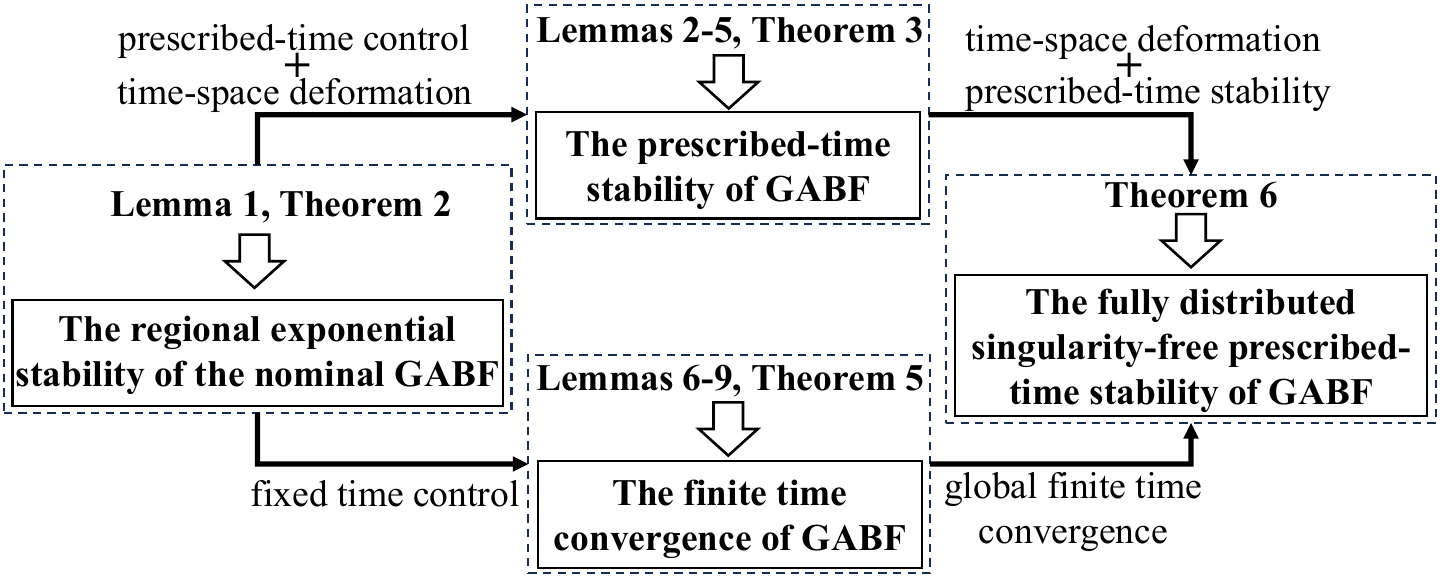}
	\caption{\textcolor{red}{Flowchart illustrating the theoretical framework and the logical relationships among the main results.}}
	\label{fig:flowchart}
\end{figure}
\subsection{The nominal generalized adaptive Bellman-Ford algorithm}
In this subsection, we demonstrate that all states in the nominal GABF will exponentially converge to (\ref{eq:bellman}), provided that all initial states are overestimates.

Define $e_i(t)$ as the error between $x_i(t)$, given in (\ref{eq:protocol}), and its stationary value $x^*_i$ defined in (\ref{eq:bellman}) as 
\begin{equation}\label{eq:error}
	e_i(t) = x_i(t) - x^*_i, \forall i \in V.
\end{equation}
It follows from \cite[Theorem~1]{mo2025} that the nominal GABF~(\ref{eq:protocol}) admits a unique global solution.
\begin{theorem}\label{the:global}
	Suppose Assumptions \ref{ass:main} and \ref{ass:Lip} hold, the solution to (\ref{eq:protocol}) is unique and exists for $[0, +\infty)$.
\end{theorem}
The following lemma shows that if the initial states of the non-source nodes are overestimates, then they remain so for all $t \in [0,+\infty)$. The proof follows from the monotonicity of $f_i$, (\ref{eq:bellman}), and the comparison inequality $\dot e_i(t)\ge -\eta e_i(t)$, and is therefore omitted.
\begin{lemma}\label{le:over}
	Suppose Assumptions \ref{ass:main}-\ref{ass:over} hold, with $e_i(t)$ defined in (\ref{eq:error}). Then $e_i(t) \geq 0$ for all $t \geq 0$ and $i \in V$.
\end{lemma}
Define the largest state error as 
\begin{equation}\label{eq:largestini}
	e_{\max}(t) = \max_{i \in V\setminus S}\{ e_i(t)\}.
\end{equation}
The nominal GABF is regionally exponentially stable, as established in the following theorem.
\begin{theorem}\label{the:exponential}
	Suppose Assumptions \ref{ass:main}-\ref{ass:over} hold, with $e_{\max}(0) $ defined in (\ref{eq:largestini}), $x_i(t)$ defined in (\ref{eq:protocol}) obeys
	\begin{equation}\label{eq:exTraj}
		x_i^* \leq 	x_i(t) \leq x^*_i + e_{\max}(0)\bar{L}(\cD(G)-1)\alpha e^{-0.5\eta t},
	\end{equation}
	where $\bar{L} =  \max\{1, L^{\cD(G) - 2} \}$, with $L$ the Lipschitz constant defined in Assumption \ref{ass:Lip}, $\cD(G)$ defined in Definition \ref{def:true}, and $\alpha = \max_{j \in \{1, \cdots, \cD(G) - 1 \}} \left( \frac{\cD(G) - 1 -j}{0.5 e} \right)^{\cD(G) - 1 - j}$.
\end{theorem}

\subsection{The prescribed-time control strategy}\label{sec:pt}
This subsection introduces our first control strategy, referred to as the prescribed-time control strategy, which eliminates the state error entirely within a prescribed time. Throughout this subsection, we define the state error $e_i(t)$ as in (\ref{eq:error}), with $x_i(t)$ following the dynamics to be formally introduced in (\ref{eq:prespecified}).

The prescribed-time control strategy is implemented by replacing the feedback gain $\eta$ in (\ref{eq:protocol}) with a time-varying scaling function $\bar{\eta}(t)$, so that (\ref{eq:protocol}) becomes 
\begin{equation}\label{eq:prespecified}
	\xd_i(t) \!=\! 
	\begin{cases}
		0, \!\!\!\!& i \in S\\
		-\bar{\eta}(t)\!\left(x_i(t) \!-\! \min_{j \in \cN(i)}\{f_i(x_j(t),w_{ij})\}\right), \!\!\!\!& i \notin S
	\end{cases},
\end{equation}
where 
\begin{equation}\label{eq:ebar}
	\bar{\eta}(t) = 
	\begin{cases}
		\gamma + 2\frac{\dot{\rho}(t)}{\rho(t)}, & t \in [0, T_p) \\
		0 , & t \geq T_p
	\end{cases},
\end{equation}
with $T_p$ a user-specified constant, $\gamma > 0$ and $\rho(t)$ obeying 
\begin{equation}\label{eq:rho}
	\rho(t) = 
	\frac{T_p^{1 + h}}{(T_p - t)^{1 + h}}, t \in [0, T_p),
\end{equation}
with $h > 0$.
The prescribed-time stability of (\ref{eq:prespecified}) is defined as follows.
\begin{definition}
	(\cite{wang2018leader,ning2022fixed}): We say (\ref{eq:prespecified}) achieves prescribed-time stability if for all $i \in V$, $x_i(t) = x_i^*$ for $t \geq T$, with $x_i^*$ defined in (\ref{eq:bellman}) and $T$ the prescribed time independent of the initial conditions.
\end{definition}
Due to the singularity of $\bar{\eta}(t)$ at time $T_p$, the global existence and uniqueness of the solution to (\ref{eq:prespecified}) cannot be established \emph{a priori}. Nevertheless, (\ref{eq:prespecified}) admits a unique solution on the interval $[0, T_p)$ by \cite[Theorem~1]{mo2025} and \cite[Theorem 3.2]{khalil}. 
\begin{lemma}\label{le:local}
	Suppose Assumption \ref{ass:main} and \ref{ass:Lip} hold. Then (\ref{eq:prespecified}) has a unique solution over $[0, T_p)$.
\end{lemma}
By the same argument as in Lemma~\ref{le:over}, the states of the non-source nodes in (\ref{eq:prespecified}) remain overestimated on $[0,T_p)$. The result follows directly from applying \cite[Lemma~6]{mo2025} to the comparison inequality $\dot e_i(t)\ge -\bar{\eta}(t)e_i(t)$, and the proof is therefore omitted.
\begin{lemma}\label{le:alloverpt}
	Suppose Assumption \ref{ass:main}-\ref{ass:over} hold. Then $e_i(t) \geq 0$ for all $i \in V$ and all $t \in [0, T_p)$.
\end{lemma}
By employing the time-space deformation approach, which squeezes the infinite time interval $[0, +\infty)$ to the prescribed one $[0, T_p)$ via a change of state variables, the following lemma demonstrates that as time approaches the prescribed time instant, the state of each non-source node converges to its stationary value. Furthermore, the upper bound on the trajectory of each non-source node over $[0, T_p)$ is also characterized.
\begin{lemma}\label{le:leftlimit}
	Suppose Assumptions \ref{ass:main}-\ref{ass:over} hold, consider (\ref{eq:prespecified}), with $x_i^*$ defined in (\ref{eq:bellman}), for all $i \in V\setminus S$, there holds
	\begin{equation}\label{eq:leftlimit}
		\lim_{t \rightarrow T_p^-} x_i(t) = x^*_i.
	\end{equation}
	In particular,  for $t \in [0, T_p)$, $x_i(t)$ obeys
	\begin{equation}\label{eq:boundt}
		x^*_i \leq	x_i(t) \leq x_i^* + Me^{-0.5 \left (\gamma t - (2+2h) \ln\frac{T_p - t}{T_p}\right )},
	\end{equation}
	where  $M  = e_{\max}(0)\bar{L}\alpha(\cD(G)-1)$, with $\alpha$ defined in Theorem \ref{the:exponential}, $e_{\max}(0)$ is defined in (\ref{eq:largestini}), and $\bar{L} = \max\{1, L^{\cD(G) - 2}\}$ with $ L$ and $\cD(G)$ defined in (\ref{eq:Lipschitz})  and Definition \ref{def:true}, respectively.
\end{lemma}
To ensure that the states of the non-source nodes remain at their stationary values at the prescribed time instant and thereafter, given that the feedback gain $\bar{\eta}(t)$ becomes a zero function beyond this time instant, we need to show that the limit of the dynamics described by (\ref{eq:prespecified}) approaches zero as time reaches the prescribed time instant. Leveraging the upper bound on the nodes' trajectories in Lemma \ref{le:leftlimit}, the following lemma further establishes the continuity of the dynamics.
\begin{lemma}\label{le:continuity}
	Suppose Assumptions \ref{ass:main}-\ref{ass:over} hold, consider (\ref{eq:prespecified}) with $T_p$ and $h$ defined in (\ref{eq:ebar})  and (\ref{eq:rho}), respectively, if $h > 0$, there holds
	\begin{equation}\label{eq:leftdot}
		\lim_{t \rightarrow T_p^-} \xd_i(t) = \ed_i(t) = 0,~ \forall i \in V.
	\end{equation}
\end{lemma}
Now we present the prescribed-time stability of (\ref{eq:prespecified}).
\begin{theorem}\label{the:existence}
	Suppose Assumptions \ref{ass:main}-\ref{ass:over} hold, consider (\ref{eq:prespecified}), with $x^*_i$ defined in (\ref{eq:bellman}), if $h > 0$, then for all $i \in V$, $x_i(t)$ is continuously differentiable on $[0, +\infty)$ and
	\begin{equation}
		x_i(t) = x^*_i, ~ \forall i \in V ~\mathrm{and}~ \forall t \geq T_p.
	\end{equation}
\end{theorem}
\begin{proof}
	From (\ref{eq:prespecified}), (\ref{eq:ebar}) and Lemma \ref{le:continuity}, $x_i(t)$ is continuous for all $i \in V$ over $[0, +\infty)$, and thus $x_i(T_p) = x^*_i$ by Lemma \ref{le:leftlimit}. Further, as $\xd_i(t) = \ed_i(t) = 0$ for all $t \geq T_p$ by (\ref{eq:ebar}), we have $x_i(t) = x^*_i$ for all $i \in V$ and $t \geq T_p.$	
\end{proof}
With $\bar{\eta}(t)$ defined in (\ref{eq:ebar}), the control input, i.e., the right-hand side of (\ref{eq:prespecified}), remains continuous and bounded even when $\bar{\eta}(t)$ grows unbounded at the preset time instant. As a result, the prescribed-time control strategy ensures a continuously differentiable trajectory of $x_i(t)$.

The theorem below concludes this section by demonstrating that the solution to (\ref{eq:prespecified}) is also unique.
\begin{theorem}
	Suppose Assumptions \ref{ass:main}-\ref{ass:over} hold, (\ref{eq:prespecified}) admits a unique solution for $t \in [0, +\infty)$.
\end{theorem}
\begin{proof}
	It follows from \cite[Theorem 1]{mo2025} that (\ref{eq:prespecified}) is Lipschitz continuous in $x_i(t)$ for all $t \in [0, t_1]$, with $t_1 < T_p$, and is continuous on $[0, t_1]$. Further, it follows from (\ref{eq:lowerli}) and (\ref{eq:ute}) in Lemma \ref{le:continuity} that there exists a $\kappa > 0$, dependent on the initial states, such that $|\xd_i(t)| \leq \kappa$ for all $t \geq 0$ and $i \in V$. From Theorem 2.40 in \cite{haddad2008nonlinear}, the solution to (\ref{eq:prespecified}) is unique.
\end{proof}
	
	

\subsection{The fixed time control strategy}
Since the overestimation requirement in Assumption~\ref{ass:over} still implicitly relies on global information about the network topology, and the infinite gain at the prescribed time instant may impede practical implementation in real-world scenarios, the remainder of this section presents a fully distributed, singularity-free prescribed-time control strategy for GABF. 

In this subsection,  we first present the fixed time control strategy for GABF. \emph{Throughout the remainder of Section \ref{sec:ptandppt}, the overestimation of initial states specified in Assumption \ref{ass:over} is no longer required.} Similarly, $e_i(t)$ used throughout this subsection denotes the state error between the desired state $x_i^*$ and the state estimate $x_i(t)$ under the fixed time control strategy defined as follows:
\begin{equation}\label{eq:singularity}
	\xd_i(t) \!=\! 
	\begin{cases}
		0, & i \in S\\
		-\eta \left ( \beta_1\xh_i(t)^{\frac{p_1}{p_2}} + \beta_2\xh_i(t)^{\frac{m_1}{m_2}} \right ), & i \notin S
	\end{cases},
\end{equation}
where $\eta, \beta_1, \beta_2 > 0$,  $m_1, m_2, p_1$ and $p_2$ are all positive odd numbers with $m_1 > m_2$ and $p_1 < p_2$, and 
\begin{equation}\label{eq:xhat}
	\xh_i(t) = x_i(t) - \min_{j \in \cN(i)}\{f_i(x_j(t),w_{ij})\}.
\end{equation}
The following lemma first establishes that all solutions to (\ref{eq:singularity}) are upper bounded.
\begin{lemma}\label{le:globalexistence}
	Suppose Assumptions \ref{ass:main} and \ref{ass:Lip} hold. Then all solutions to (\ref{eq:singularity}) are upper bounded over the maximal interval of existence.
\end{lemma}
To establish a lower bound for the solutions to (\ref{eq:singularity}), we first partition the set of non-source nodes based on the ordering of their optimal values as characterized by (\ref{eq:bellman}).
\begin{definition}\label{def:S}
	Let $\bar{x}_l^*$ denote the $l$-th smallest value in the set $\{x_i^* \}_{i \in V\setminus S}$, with $l \in \{1, 2, \cdots, R\}$. Note that $1 \leq R \leq |V\setminus S|$, and $R < |V\setminus S|$ may hold as multiple non-source nodes can share the same stationary value. Define 
	\begin{equation}\label{eq:S}
		\cS_l = \{i \in V\setminus S| x_i^* = \xb_l^* \}, l \in \{1, 2, \cdots, R\}
	\end{equation}
	as the set of non-source nodes whose stationary value equals $\bar{x}_l^*$. Specifically, define $\cS_0 := S$. Furthermore, let 
	\begin{equation}\label{eq:small}
		x_{\min}^{(l)}(t) = \min_{i \in V \setminus \cup_{j = 0}^{l} \cS_j} x_i(t), l \in \{0,1,\cdots, R-1 \}
	\end{equation}
	be the smallest state among nodes not in $\cup_{j = 0}^{l}\cS_j$. The set of nodes attaining $x_{\min}^{(l)}(t)$ at time $t$ is defined as
	\begin{equation}\label{eq:kt}
		\cK^l(t) = \{ i \in V\setminus \cup_{j = 0}^{l} \cS_j ~|~ x_i(t) = x_{\min}^{(l)}(t)  \}.
	\end{equation}
\end{definition}

The lower boundedness of all solutions to (\ref{eq:singularity}) is established as follows. 
\begin{lemma}\label{le:globalexistence1}
	Suppose Assumptions \ref{ass:main} and \ref{ass:Lip} hold. Then all solutions to (\ref{eq:singularity}) are lower bounded over the maximal interval of existence.
\end{lemma}
Then we have the following appealing result (\cite[Proposition 2.1]{bhat2000finite}. 
\begin{lemma}\label{le:globalbound}
	Suppose Assumptions \ref{ass:main} and \ref{ass:Lip} hold. The solutions to (\ref{eq:singularity}) are bounded and defined over $[0, +\infty)$. In particular, for all $i \in V\setminus S$, $x_i(t) \geq \min\{\xb_1^\ast, x^{(0)}_{\min}(0)\}$ for all $t \geq 0$, with $\xb_1^\ast$ and $x^{(0)}_{\min}(0)$ defined in Definition \ref{def:S}.
\end{lemma}
The following lemma shows that non-source nodes converge to their optima in finite time, provided that all nodes with smaller optima have already converged.
\begin{lemma}\label{le:overfinite}
	Suppose Assumptions \ref{ass:main} and \ref{ass:Lip} hold, with $\xb_l^*, \cS_l, x_{\min}^{(l)}(t)$, $\cK^l(t)$ and $l \in \{1,\cdots, R\}$ defined in Definition \ref{def:S}. If for all $i \in \cup_{j \in \{0, 1, \cdots, l \}}\cS_i$ with $0 \leq l \leq R - 1$, we have $x_i(t) = x_i^*$ for $t \geq t_1$, then for all $i \in \cS_{l+1}$,
	\begin{equation}
		x_i(t) = x_i^*, ~
		\forall t \geq t_1 + 2T_* + \max\{T_f,0\},
	\end{equation}
	where $T_*$ and $T_f$ satisfy
	\begin{equation}\label{eq:fixtime}
	T_* = 	 \frac{1}{\eta\beta_1}\frac{m_2}{m_1 - m_2} + \frac{1}{\eta\beta_2}\frac{p_1}{p_2 - p_1},
	\end{equation}
	\begin{equation}\label{eq:Tf}
		T_f \leq  \frac{\xb_R^* - \zeta - \min\{\xb_1^*, x^{(0)}_{\min}(0)\}}{\eta\left (\beta_1\zeta^{\frac{p_1}{p_2}} + \beta_2\zeta^{\frac{m_1}{m_2}} \right )},
	\end{equation}
	with $\zeta$ defined in (\ref{eq:progressive}), $\beta_1,\beta_2,p_1,p_2,m_1,m_2$ and $\eta$ given in (\ref{eq:singularity}).
\end{lemma}
As all source nodes' states are anchored at their stationary values, the following theorem shows that (\ref{eq:singularity}) will converge after a finite time, thereby establishing the global finite time stability of (\ref{eq:singularity}).
\begin{theorem}\label{the:fix}
	Suppose Assumptions \ref{ass:main} and \ref{ass:Lip} hold. With $x_i^*, T_f$ and $T_*$ defined in (\ref{eq:bellman}) and Lemma \ref{le:overfinite}, for all $i \in V$,
	\begin{equation}
		x_i(t) = x_i^*, ~ \forall t \geq 2|V\setminus S|T_* + |V\setminus S|\max\{0, T_f\}.
	\end{equation}
\end{theorem}

\subsection{The fully distributed singularity-free prescribed-time control strategy}\label{sec:free}
This subsection establishes the fully distributed singularity-free prescribed-time stabilization of GABF, leveraging the fixed time control strategy derived above. Throughout this subsection, $e_i(t)$ is defined as the difference between $x_i^*$ and the state under the fully distributed singularity-free prescribed-time control strategy:
\begin{flalign}\label{eq:fully}
	&\!\!\!\!	\xd_i(t) =  \nonumber \\
	&\!\!\!\!	\begin{cases}
		0, \!\!\!\!& i \in S,\\
		-\bar{\eta}(t) \left ( \beta_1\xh_i(t)^{\frac{p_1}{p_2}} + \beta_2\xh_i(t)^{\frac{m_1}{m_2}} \right ), \!\!\!\!& i \notin S~ \&~ \xh_i(t) \neq 0, \\
		-\eta ( \beta_1\xh_i(t)^{\frac{p_1}{p_2}} + \beta_2\xh_i(t)^{\frac{m_1}{m_2}} ), \!\!\!\!& i \notin S ~\&~ \xh_i(t) = 0,
	\end{cases}
\end{flalign}
where $\bar{\eta}(t)$ is defined in (\ref{eq:ebar}), $\eta, \beta_1, \beta_2 > 0$,  $m_1, m_2, p_1$ and $p_2$ are all positive odd numbers with $m_1 > m_2$ and $p_1 < p_2$, and $\xh_i(t)$ is given by (\ref{eq:xhat}).

The main result of this subsection is summarized as follows, which ensures convergence to the stationary value while preventing the gain from diverging at $t = T_p$ in the second bullet of (\ref{eq:fully}).
\begin{theorem}\label{the:singularity}
	Suppose Assumptions \ref{ass:main} and \ref{ass:Lip} hold, consider (\ref{eq:fully}), with $x_i^*$ defined in (\ref{eq:bellman}), for all $i \in V$, 
	\begin{equation}\label{eq:T*P}
		x_i(t) = x^*_i,~ \forall t \geq T^*_p,
	\end{equation}
	where $T^*_p < T_p$, with $T_p$ the prescribed time in (\ref{eq:rho}).
\end{theorem}
\begin{remark}\label{re:implementation}
(\ref{eq:fully}) and (\ref{eq:prespecified}) can be discretized and implemented on digital hardware using standard ODE solvers. Due to the non-autonomous and stiff nature of the system induced by \(\bar{\eta}(t)\) defined in (\ref{eq:ebar}), implicit numerical integration methods, such as the backward Euler method, are generally preferred. Although explicit methods may fail near the prescribed time \(T_p\) even with very small time steps, explicit solvers such as MATLAB's \texttt{ode45} remain practical, as convergence typically occurs well before \(T_p\).
\end{remark}

\section{Simulations}\label{sec:simulations}
In this section, we empirically verify the theoretical results presented in the previous section and demonstrate the practicality of GABF, through simulations based on real-world application scenarios.
\subsection{Comparison of the prescribed-time control strategies}\label{sec:comparison}
\begin{figure}
	\centering
	\subfigure[prescribed-time control strategy]{  
		\includegraphics[width = 0.9\columnwidth]{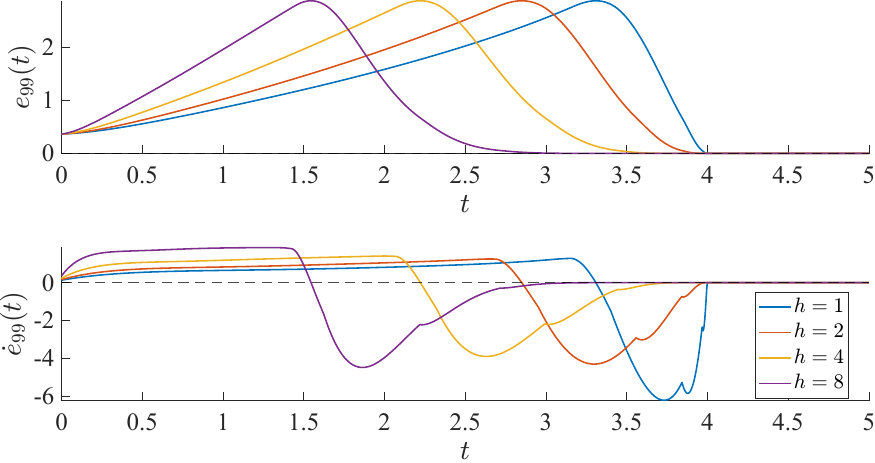}}
	\subfigure[fully distributed singularity-free prescribed-time control strategy]{  
		\includegraphics[width = 0.9\columnwidth]{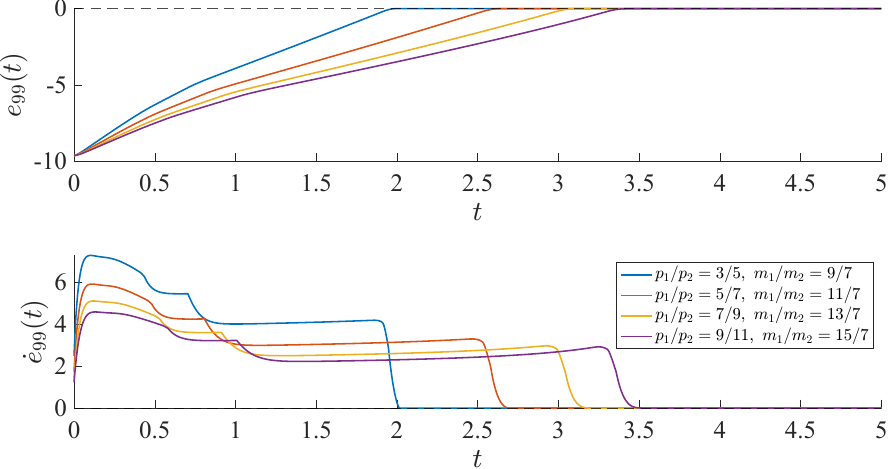}}
	\caption{Performance of different prescribed-time control strategies.}
	\label{fig:comparison}
\end{figure}
We consider a directed graph with 500 randomly distributed nodes (499 non-source nodes and 1 source node) in a $10 \times 10 ~\mathrm{km}^2$ area. Each node has on average 10 out-neighbors, and each non-source node has at least one directed path to the source. Let $f_i(a, b) = a + b$ for $i \in V\setminus S$, where GABF reduces to DBMC. The Euclidean distance between two nodes is used as their edge weight. 

The prescribed-time control strategy (\ref{eq:prespecified}) is first applied to GABF. We set $\gamma=1$, $T_p=4$, and vary $h \in \{1,2,4,8\}$. The initial conditions are $x_i(0)=0$ for $i \in S$ and $x_i(0)=10$ for all $i \in V \setminus S$, yielding overestimated initial states. Figure \ref{fig:comparison}(a) shows the state error and control input of node 99. It is observed that all state errors converge within the prescribed time $T_p=4$, while all the control inputs remain bounded. Moreover, larger $h$ leads to faster convergence as the control effort increases.

Then the fully distributed singularity-free prescribed-time control strategy (\ref{eq:fully}) is applied to GABF. We set $\beta_1=\beta_2=1$, $h=1$, $\gamma=10$, $T_p=4$, $\eta=15$, and vary $\frac{p_1}{p_2} \in \{\frac{3}{5}, \frac{5}{7}, \frac{7}{9}, \frac{9}{11}\}$ and $\frac{m_1}{m_2} \in \{\frac{9}{7}, \frac{11}{7}, \frac{13}{7}, \frac{15}{7}\}$. The initial states cover both underestimates and overestimates. As shown in Figure \ref{fig:comparison}(b), the state error converges before the prescribed time $T_p=4$, while all the control inputs remain bounded. Moreover, smaller $\frac{p_1}{p_2}$ and $\frac{m_1}{m_2}$ yield faster convergence at the cost of increased control effort.

\subsection{Application in robotic manipulator path planning}
In this subsection, we explore the application of GABF to path planning for industrial robotic manipulators, where both distance cost and manipulability cost are taken into account \cite{shen2023adaptive}. In this context, $f_i(x_j(t), w_{ij})$ in (\ref{eq:protocol}) is defined as
\begin{equation}\label{eq:finew}
	f_i(x_j(t), w_{ij}) = x_j(t) + (1 - \beta)\omega(j) + \beta w_{ij},
\end{equation}
where $\omega(j)$ represents the manipulability cost depending on the configuration at the position of node $j$, $w_{ij}$ represents the Euclidean distance between nodes $i$ and $j$, $x_j(t)$ represents the estimated optimal cost in terms of both the path length and manipulability measure, and $\beta \in [0, 1]$ sets a tradeoff between the distance cost and the manipulability cost. It can be verified that the kernel function $f_i: \mathbb{R} \times \mathbb{R}_{> 0} \rightarrow \mathbb{R}$ described above satisfies properties of progressiveness and strict monotonicity as required. In particular, the manipulability cost $\omega(j)$ 
is pre-determined using RRT algorithm to avoid obstacles \cite{kuffner2000rrt}.
\begin{figure}
	\centering
	\subfigure[path planning using GABF and RRT]{  
		\includegraphics[width = 0.7\columnwidth]{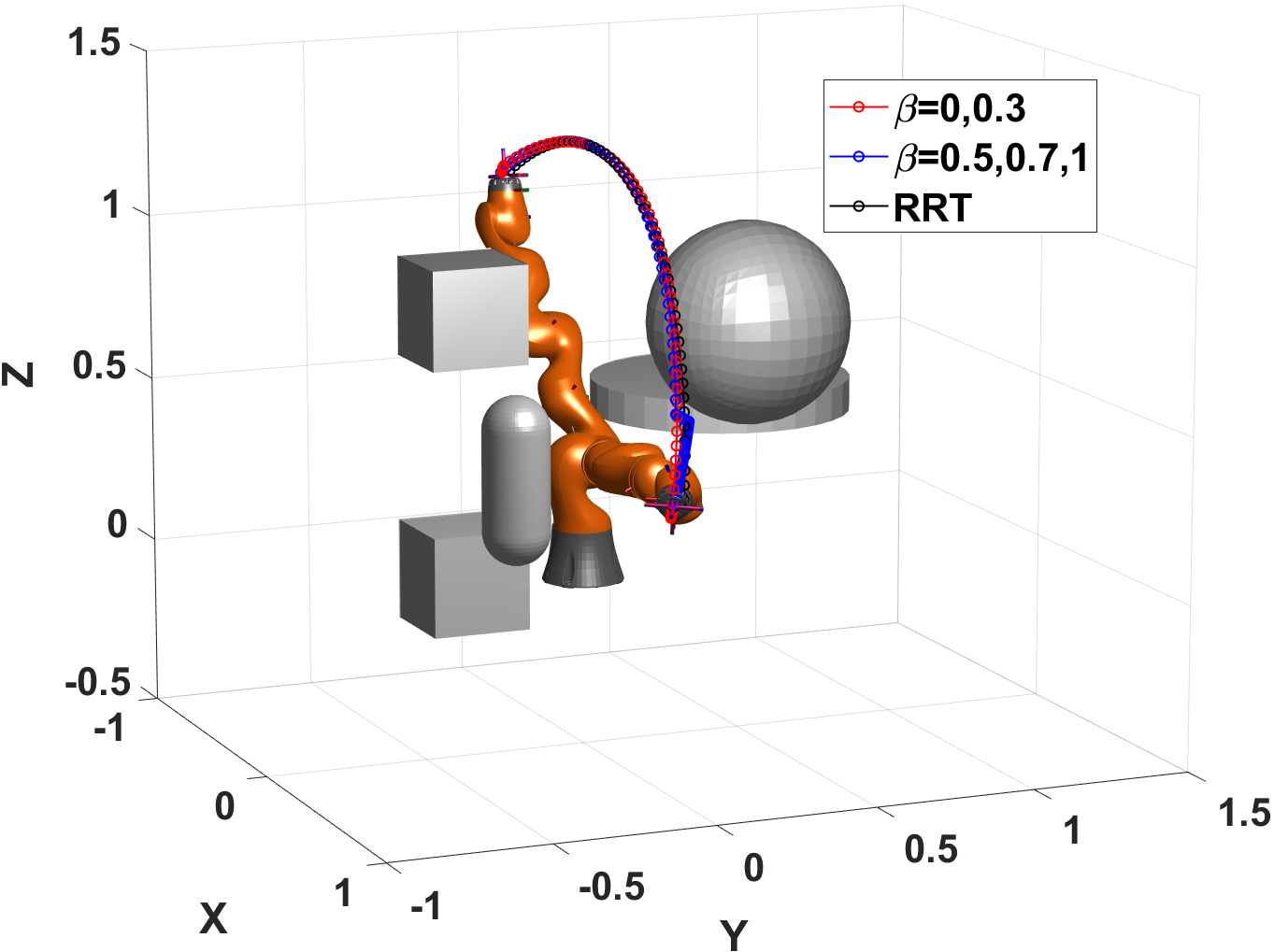}}
	\subfigure[enlarged view of the trajectories]{  
		\includegraphics[width = 0.7\columnwidth]{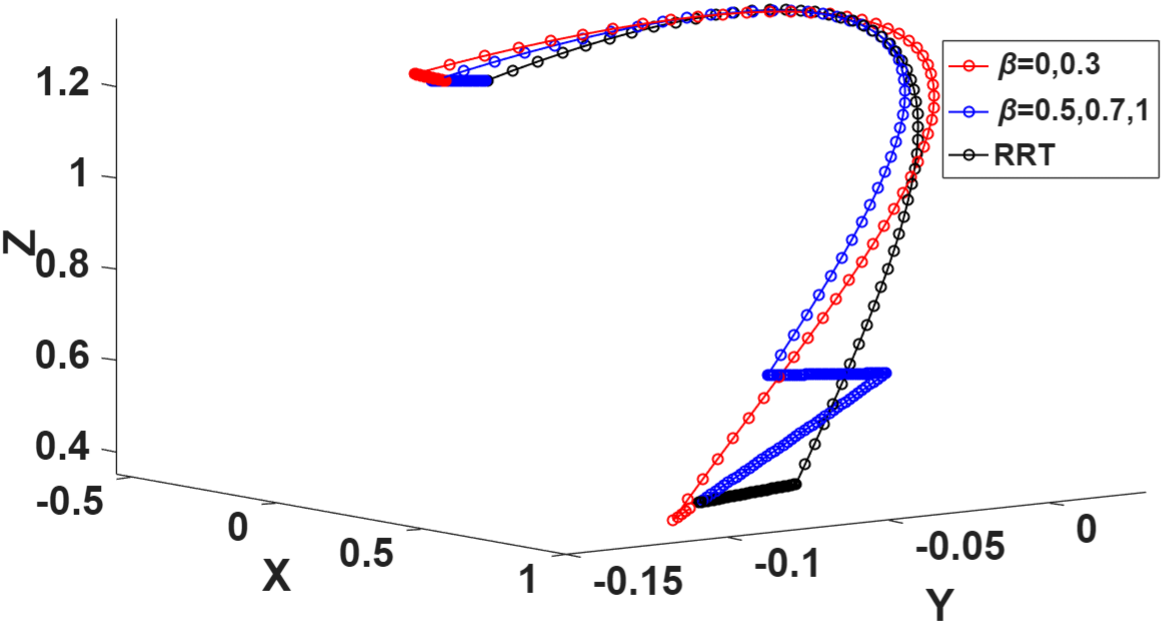}}
	\caption{Comparison of path planning for a robotic manipulator using RRT and GABF.}
	\label{fig:robotic}
\end{figure}
Figure \ref{fig:robotic}(a) demonstrates the results of applying GABF to a KUKA LBR iiwa 14 R820 7-axis manipulator in an environment populated with several obstacles, using the fully distributed singularity-free prescribed-time control strategy (\ref{eq:fully}), with parameters consistent with those in Section \ref{sec:comparison}. Note that, to ensure numerical consistency and comparability between the distance and manipulability costs, both cost terms are normalized prior to aggregation. 

We generate a network consisting of 1012 nodes and 2029 edges, along with the corresponding configuration at each node using the RRT algorithm. GABF is then applied to plan a motion trajectory from an initial location at $(0.92, -0.10, 0.42)$ to a destination at $(-0.52, -0.04, 1.13)$, with $\beta$ in (\ref{eq:finew}) increasing from 0 to 1. It can be observed from Figure \ref{fig:robotic}(a) that GABF produces the same path for both $\beta = 0$ and $\beta = 0.3$, as well as a common path for $\beta = 0.5, 0.7$ and 1. Furthermore, from Figure \ref{fig:robotic}(b), most segments of the three trajectories differ from each other, and the trajectory generated with $\beta = 0.5, 0.7,$ or $1$ is closer to that produced by the RRT implementation, since in this case the distance cost gradually outweighs the manipulability cost, while the RRT algorithm only considers the distance cost. \textcolor{red}{From (\ref{eq:finew}), the red path obtained with $\beta=0$ or $0.3$ prioritizes manipulability, whereas the blue path obtained with $\beta=0.5$, $0.7$, or $1$ emphasizes path length. Consequently, the red path achieves a significantly lower accumulated manipulability cost (0.0331 versus 1.1038), at the expense of a slightly larger accumulated distance cost (47.4037 versus 46.2546). As can be observed from Figure \ref{fig:robotic}(b), although the red path does not appear significantly longer than the blue path, it exhibits larger curvature in the three-dimensional workspace, resulting in a longer actual travel distance. It is worth noting that the blue path exhibits a zigzag pattern. This is mainly due to the non-uniform sampling property of the RRT algorithm, which results in an irregular distribution of nodes and edges in the generated graph. These numerical results are fully consistent with the corresponding objective functions: for $\beta=0$ or $0.3$, the total cost of the red path is lower than that of the blue path, whereas the opposite holds for $\beta=0.5$, $0.7$, and $1$.}

\subsection{Learning-based path planning}
In the following, we investigate the performance of GABF using (\ref{eq:fully}) in a learning-based path finding scenario, as described in \cite{kulvicius2023combining}. In this case, $\xh_i(t)$ in (\ref{eq:fully}) becomes $\xh_i(t) = x_i(t) - \min_{j \in \cN(i)}\{1 -(1-x_j(t))*\bar{w}_{ij} \}$, with $\bar{w}_{ij} = 1 - \frac{w_{ij}}{K}$. It has been shown in \cite{kulvicius2023combining} that the resulting path under (\ref{eq:fully}), using $\xh_i(t)$ above, remains the shortest path for node $i$ when $K$ is sufficiently large, and the newly defined function $f_i: \mathbb{R}_{[0,1)} \times \mathbb{R}_{> 0} \rightarrow \mathbb{R}_{(0,1)}$ given by 
\begin{equation}\label{eq:wb}
	f_i(x_j(t), w_{ij}) = 1 -(1-x_j(t))\bar{w}_{ij}
\end{equation}
satisfies the progressive and monotonicity properties. The merit of utilizing the updated $f_i$ lies in its adaptability: in scenarios where edge weights are unknown or vary over time, $x_i(t)$ computed using (\ref{eq:wb}), rather than using $f_i(x_j(t), w_{ij}) = x_j(t) + w_{ij}$ as in DBMC, can be directly incorporated into the learning process of $\bar{w}_{ij}$, and the resulting $\bar{w}_{ij}$ can, in turn, facilitate the computation of the shortest path. In this way, the repeated transformation between path finding costs and network learning activations can be avoided. 

We adopt the Hebbian-type synaptic
learning rule \cite{kusmierz2017learning} to modify $\bar{w}_{ij}$ in (\ref{eq:wb}). Specifically, during each training iteration, the change in each $\bar{w}_{ij}$ follows the rule: $\Delta \bar{w}_{ij} = \alpha (1 - x_i(t))(1 - x_j(t)) R$,
where $\alpha > 0$ is the learning rate, the reward signal $R$ is defined as $R = 1 - \beta d_{ij}$ with $\beta > 0$ and $d_{ij}$ representing the information regarding the Euclidean distance between nodes $i$ and $j$, i.e., smaller distances between locations lead to larger reward values and larger weight updates. While \cite{kulvicius2023combining} employs a modified version of the classic Bellman-Ford algorithm for the shortest path computation, which requires prior knowledge of the graph, the prescribed-time stability achieved by the fully distributed singularity-free prescribed-time control strategy ensures convergence within a user-defined time, in a knowledge-free sense. Furthermore, the distributed nature of GABF enables simultaneous training of each $\bar{w}_{ij}$, in contrast to the sequential training approach used in \cite{kulvicius2023combining}.
\begin{figure}
	\centering
	\subfigure[path planning after training iteration 1]{  
		\includegraphics[width = 0.7\columnwidth]{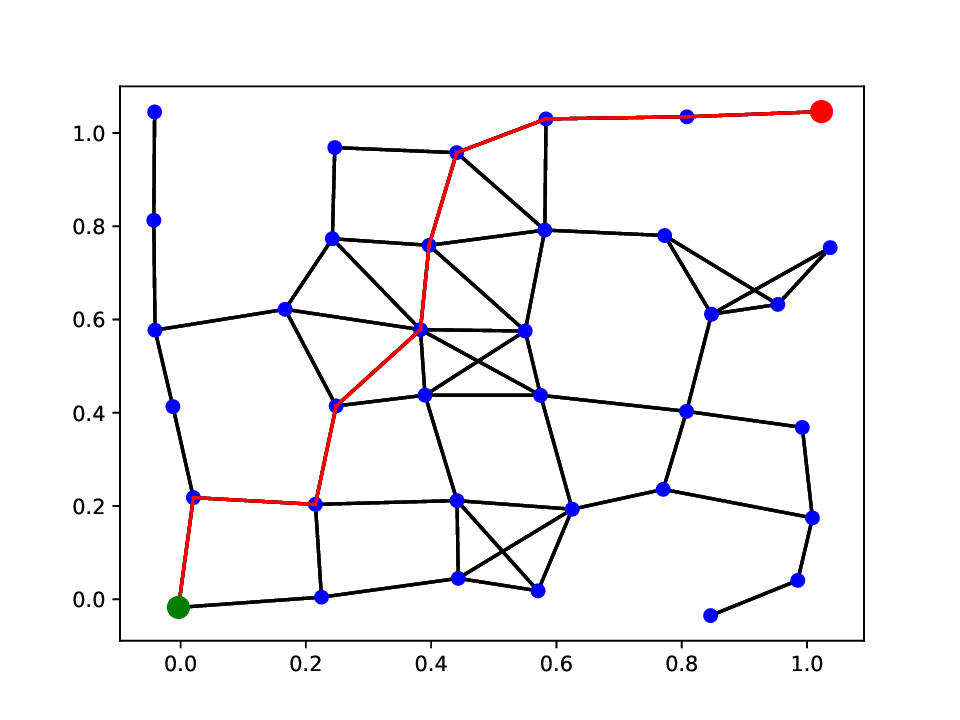}}
	\subfigure[path planning after training iteration 3]{  
		\includegraphics[width = 0.7\columnwidth]{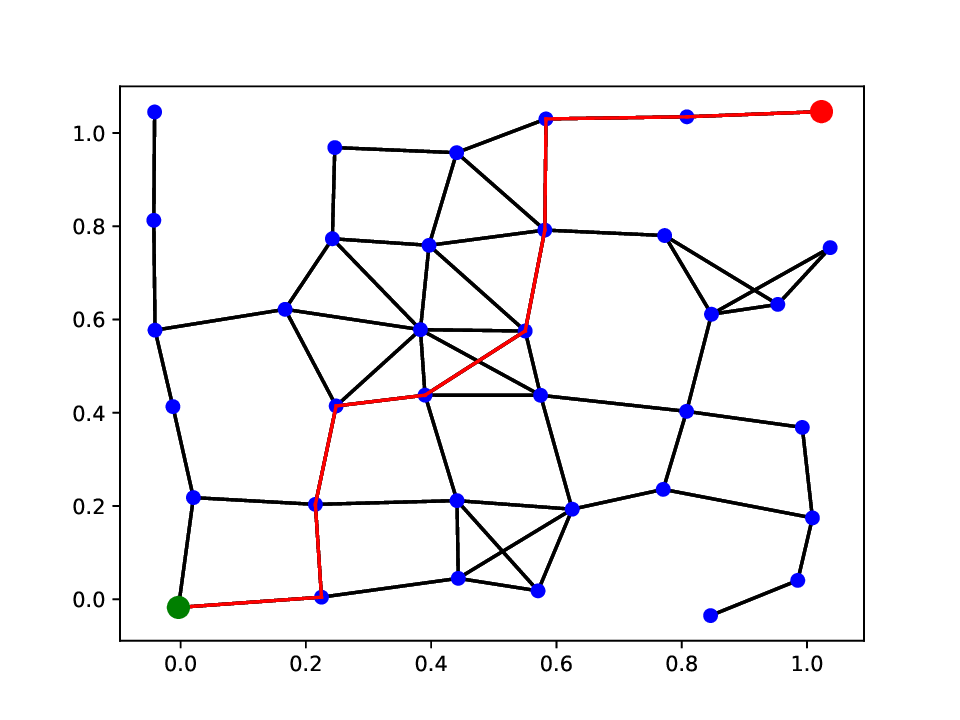}}
	\subfigure[path planning after training iteration 6]{  
		\includegraphics[width = 0.7\columnwidth]{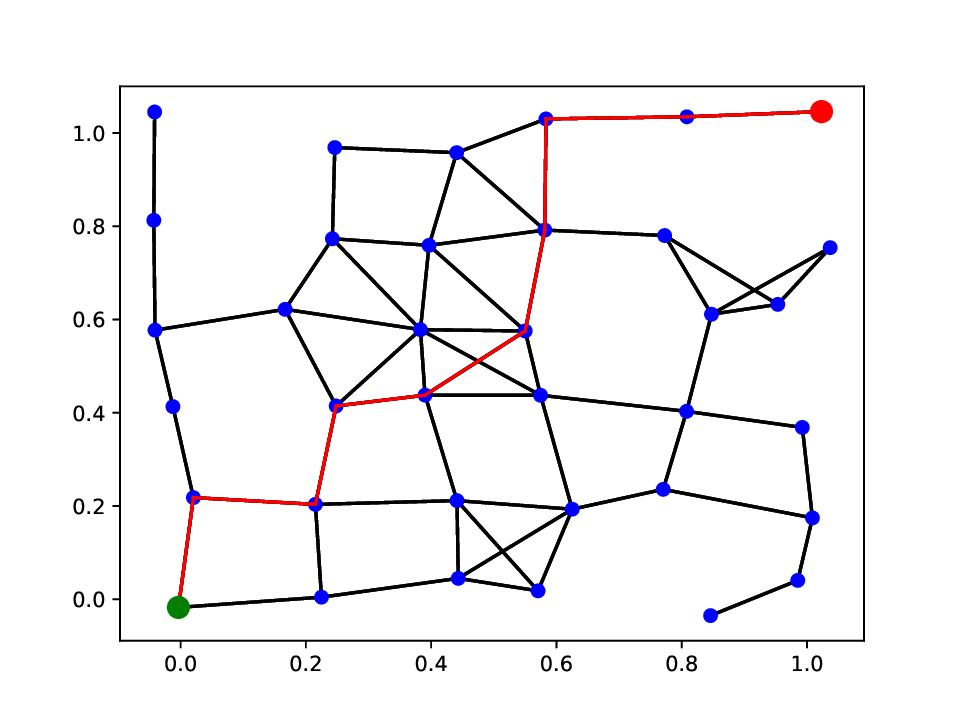}}
	\caption{Path planning using the Hebbian-type synaptic learning rule and GABF under the fully distributed
		singularity-free prescribed-time control strategy. The red trajectory represents the resulting path.}
	\label{fig:learning}
\end{figure}

Figure \ref{fig:learning} illustrates the simulation results. In this scenario, 36 nodes are uniformly distributed over a $[0,1] \times [0,1]$ area, with each node's coordinates randomly perturbed by an offset drawn uniformly from $[-0.05, 0.05]$. Nodes can communicate with their neighbors within a communication range of 0.25. The green node at the bottom-left corner of Figure \ref{fig:learning} represents the start point, while the red node at the top-right marks the destination. The initial value of $\bar{w}_{ij}$ is set to 0.5. After the initial training, GABF produces a feasible but suboptimal path during the path planning stage, as the trained weights $\bar{w}_{ij}$, which approximate $1 - \frac{w_{ij}}{K}$, has not converged yet. After two training iterations, GABF generates an improved but still non-optimal path. After the sixth training iteration, GABF successfully identifies and converges to the shortest path, as illustrated in Figure \ref{fig:learning}(c).

\section{Conclusions}\label{sec:con}
This paper presents two control strategies for GABF. The first control strategy enables GABF to converge to the stationary value within a prescribed time, while the adopted time-varying scaling function remains independent of the underlying graph parameters. The second control strategy further eliminates the requirement for overestimated initial states and ensures convergence strictly before the prescribed time, thereby avoiding the occurrence of infinite gain at the terminal time. 
While these results extend the asymptotic and finite time convergence outcomes for GABF \cite{VABDP-ACM,mo2021global} as well as those developed for DBMC \cite{mo2025,Zhang2017,zhang2017perturbing}, they also highlight a distinct aspect of the tradeoff between convergence rate and control complexity: the prescribed-time control strategy requires overestimated initial states, whereas the fully distributed singularity-free prescribed-time control strategy incurs higher computational complexity due to the use of fractional powers. \textcolor{red}{Moreover, the algorithm studied in this paper is a specific Bellman-operator-based algorithm. Existing studies on such Bellman operators have primarily concentrated on discrete-time value iteration in dynamic programming, where the main objective is to establish geometric convergence under repeated applications of the Bellman operator \cite{bertsekas2012dynamic}. By contrast, prescribed-time control has mainly been developed for continuous-time dynamical systems described by ordinary differential equations. Developing a general prescribed-time stabilization theory for Bellman-operator-induced mappings that bridges these two research directions remains an interesting and challenging topic, which we leave for future research.}


\bibliographystyle{automaticanew} 
\bibliography{refs}

\section*{Appendix}\label{App}
\noindent

\setlength{\parskip}{0pt}

{\bf Proof of Theorem \ref{the:exponential}:} Consider a sequence of nodes $i_0, \cdots, i_T$ such that $i_0 \in S$ and $i_{\ell} \in \cP(i_{\ell + 1})$ for $\ell \in \{0, \cdots, T-1\}$. All nodes are in such a sequence and $T \leq \cD(G) - 1$ by definition. We prove by induction that 
\begin{equation}\label{eq:induction}
	e_{i_\ell}(t)\! \leq \!e^{-\eta t}\sum_{j = 1}^{\ell}e_{i_j}(0)\frac{(L\eta t)^{\ell - j}}{(\ell - j)!}, \forall \ell \in \{1, \cdots, T\}.
\end{equation}

For $i_1$, it follows from (\ref{eq:protocol}) and (\ref{eq:error}) that 
\begin{flalign}
	&\ed_{i_1}(t) =	\xd_{i_1}(t) = -\eta(x_{i_1}(t) - \min_{j \in \cN(i_1)}\{f_{i_1}(x_j(t),w_{i_1j})\})  \nonumber\\
	&\leq  -\eta(x_{i_1}(t) - f_{i_1}(x^*_{i_0}, w_{i_1i_0})) =  -\eta e_{i_1}(t), \label{eq:tc}
\end{flalign}
where (\ref{eq:tc}) uses the fact that $i_0 \in \cP(i_1)\cap S$.
Then $e_{i_1}(t) \leq e_{i_1}(0)e^{-\eta t}$ and (\ref{eq:induction}) holds for $\ell = 1$. Suppose (\ref{eq:induction}) holds for some $\ell \in \{1,\cdots, T-1\}$. For node $i_{\ell + 1}$, it follows from (\ref{eq:protocol}) that
\begin{flalign}
	&	\ed_{i_{\ell + 1}}(t)  \leq -\eta(x_{i_{\ell + 1}}(t) - f_{i_{\ell + 1}}(x^*_{i_\ell},w_{i_{\ell + 1}i_\ell})  \nonumber \\
	& ~~~ + f_{i_{\ell + 1}}(x^*_{i_\ell},w_{i_{\ell + 1}i_\ell}) - f_{i_{\ell + 1}}(x_{i_\ell}(t),w_{i_{\ell + 1}i_\ell})   ) \label{eq:co}\\
	&\leq  -\eta(e_{i_{\ell+1}}(t) - Le_{i_{\ell}}(t) ) \label{eq:lip}\\
	&\leq -\eta\Big(e_{i_{\ell+1}}(t) - Le^{-\eta t}\sum_{j = 1}^{\ell}L^{\ell - j}e_{i_j}(0)\frac{(\eta t)^{\ell - j}}{(\ell - j)!}\Big ), \label{eq:uin}
\end{flalign}
where (\ref{eq:co}) follows from $i_\ell \in \cN(i_{\ell+1})$, (\ref{eq:lip}) uses (\ref{eq:bellman}) and $i_\ell \in \cP(i_{\ell + 1})$,  as well as (\ref{eq:Lipschitz}), $f_i$ is strictly monotonically increasing w.r.t. its first argument and $x_{i_\ell}(t) \geq x^*_{i_\ell}$ as established in Lemma \ref{le:over}, and (\ref{eq:uin}) uses our induction hypothesis. By the comparison principle, (\ref{eq:uin}) leads to $e_{i_{\ell + 1}}(t) \leq e^{-\eta t}\sum_{j = 1}^{\ell + 1}L^{\ell + 1 - j}e_{i_j}(0)\frac{(\eta t)^{\ell+1- j}}{(\ell+1 - j)!}$, and thus (\ref{eq:induction}) holds. As $T \leq \cD(G) - 1$, for all $i \in 
V\setminus S$, there holds
\begin{flalign}\label{eq:errorbound}
	 e_i(t) \leq  e_{\max}(0)e^{-\eta t} \sum_{j = 1}^{\cD(G) - 1} \frac{(L\eta t)^{\cD(G) - 1- j}}{(\cD(G) - 1 - j)!}.
\end{flalign}
It can be readily verified that 
(\ref{eq:errorbound}) can be written as 
\begin{flalign}
	&	e_i(t) \leq e^{-\eta t}\bar{L}e_{\max}(0)\sum_{j = 1}^{\cD(G) - 1}(\eta t)^{\cD(G) - 1 - j}\label{eq:gen} \\
	& \leq 	e_{\max}(0)\bar{L}(\cD(G)-1)\alpha e^{-0.5\eta t},\label{eq:dup}
\end{flalign}
where in (\ref{eq:gen}) $\bar{L} = \max\{1, L^{\cD(G) - 2}  \}$, and in (\ref{eq:dup}) $\alpha = \max_{j \in \{1, \cdots, \cD(G) - 1 \}} \left( \frac{\cD(G) - 1 -j}{0.5 e} \right)^{\cD(G) - 1 - j}$. As $e_i(t) \geq 0$ for all $t \geq 0$ and $i \in V\setminus S$ by Lemma \ref{le:over}, our claim follows.



{\bf Proof of Lemma \ref{le:leftlimit}:} Define $\tau: [0, T_p) \rightarrow [0, +\infty)$ as
\begin{equation}\label{eq:ts}
	\tau(t) = \gamma t - 2(1+h)\ln\frac{T_p - t}{T_p},
\end{equation} 
with $\gamma$ and $h$ defined in (\ref{eq:ebar}) and (\ref{eq:rho}), respectively. $\tau(t)$ is continuous and strictly increasing over $[0, T_p)$, and $\dot{\tau}(t) = \gamma + 2\frac{1+h}{T_p - t} = \bar{\eta}(t)$. Define
\begin{equation}\label{eq:nstate}
	y_i(\tau) = x_i(t(\tau)), \forall i \in V,
\end{equation}
with $x_i(t)$ defined in (\ref{eq:prespecified}). It follows that for $\tau \in [0, +\infty)$
\begin{flalign}
	&	\frac{dy_i(\tau)}{d\tau} \!=\! \frac{dx_i(t)}{dt}\frac{dt}{d\tau} \!= \! \min_{j \in \cN(i)}\{f_i(x_j(t),w_{ij})\} \!-\! x_i(t)  \label{eq:range} \\
	& = -(y_i(\tau) - \min_{j \in \cN(i)}\{f_i(y_j(\tau),w_{ij})\}), \label{eq:rep}
\end{flalign}
where (\ref{eq:range}) uses $\frac{dt}{d\tau} = 1/\bar{\eta}(t)$, and (\ref{eq:rep}) uses (\ref{eq:nstate}). Therefore, $y_i(\tau)$ satisfies the nominal GABF in (\ref{eq:protocol}) with $\eta = 1$. It then follows from Theorem~\ref{the:exponential} that $y_i(\tau)$ converges exponentially to $x_i^*$ for all $i \in V$. In particular, it follows from (\ref{eq:exTraj}) that $y_i(\tau)$ obeys
\begin{flalign}
	y_i(\tau) \leq x^*_i + e_{\max}(0)\bar{L}(\cD(G)-1)\alpha e^{-0.5 \tau}, \label{eq:tran}
\end{flalign}
where $\alpha$ and $\bar{L}$ are defined in (\ref{eq:dup}). Putting (\ref{eq:tran}) into (\ref{eq:nstate}), with Lemma \ref{le:alloverpt} and $M = e_{\max}(0)\bar{L}\alpha(\cD(G)-1)$, we obtain that for all $i \in V\setminus S$ and $t \in [0, T_p)$
\begin{equation*}
	x_i^* \leq	x_i(t) \leq x_i^* + Me^{-0.5 (\gamma t - (2+2h)\ln\frac{T_p - t}{T_p})},
\end{equation*}
{\bf Proof of Lemma \ref{le:continuity}:}
(\ref{eq:leftdot}) holds trivially for $i \in S$.  For $i \in V\setminus S$, it follows from (\ref{eq:prespecified}) that for $t \in [0, T_p)$
\begin{flalign}
	&\!\!\!\!	\xd_i(t) =  \ed_i(t) = \bar{\eta}(t)(\min_{j \in \cN(i)}\{f_i(x_j(t),w_{ij})\} - x_i(t)) \nonumber \\
	&\!\!\!\! \geq \bar{\eta}(t)(\min_{j \in \cN(i)}\{f_i(x^*_j,w_{ij})\} - x_i(t) )  = -\bar{\eta}(t)e_i(t), \label{eq:stu}
\end{flalign}
where the inequality in (\ref{eq:stu}) follows that  $x_j(t) \geq x^*_j$ for all $t \in [0, T_p)$ and $j \in V$, and from that $f_i$ is monotonically increasing in its first argument, and the equality in (\ref{eq:stu}) uses (\ref{eq:bellman}). Further, from (\ref{eq:boundt}) and (\ref{eq:ebar}), (\ref{eq:stu}) yields
\begin{eqnarray}
	\ed_i(t) 
	&\geq& -Me^{-0.5 \gamma t }(\gamma + \frac{2+2h}{T_p - t})(\frac{T_p - t}{T_p})^{1+h} \label{eq:lowerli} . 
\end{eqnarray}
Therefore, we have $\lim_{t \rightarrow T_p^-} \ed_i(t) \geq 0$ if $h > 0$.

Note that (\ref{eq:stu}) also implies that $e_i(t)\geq \frac{1}{\rho^2(t)}e^{-\lambda t}e_i(0)$ by \cite[Lemma 6]{mo2025} and comparison principle. Again, let $j \in \cP(i)$, it follows from (\ref{eq:prespecified}) that for $t \in [0, T_p)$
\begin{flalign}
	& \ed_{i}(t) \leq  -\bar{\eta}(t) \big(  x_{i}(t) -  f_i(x_j(t),w_{ij}) +  f_i(x^*_j,w_{ij}) \nonumber\\
	&  ~~~~~~~~~~~- f_i(x^*_j,w_{ij})  \big) \nonumber \\
	& \leq -\bar{\eta}(t)(e_{i}(t) - Le_{j}(t)) \label{eq:utt}  \\
	& \leq -\bar{\eta}(t)(\frac{e^{-\lambda t}e_i(0)}{\rho^2(t)} \!-\! LMe^{-0.5 (\gamma t - (2+2h) \ln\frac{T_p - t}{T_p})}  ),  \label{eq:ute}
\end{flalign}
where (\ref{eq:utt}) uses (\ref{eq:bellman}) and (\ref{eq:Lipschitz}), and (\ref{eq:ute}) uses (\ref{eq:boundt}) in Lemma \ref{le:leftlimit}.
Therefore, it follows from (\ref{eq:ebar}) and (\ref{eq:lowerli}) that $\lim_{t \rightarrow T_p^-} \ed_i(t) \leq 0$ if $h > 0$, and our result follows.

{\bf Proof of Lemma \ref{le:globalexistence}:} 
Let $x_i(t): [0, \tau_{\max}) \rightarrow \mathbb{R}$ be the right maximally defined solution to (\ref{eq:singularity}) for $i \in V\setminus S$. We prove by induction that $x_i(t)$ is upper bounded over $[0, \tau_{\max})$. Consider any $i \in \cF_1$,  from (\ref{eq:singularity}), 
\begin{flalign}
	& \ed_i(t) =	\xd_i(t) = -\eta ( \beta_1\xh_i(t)^{\frac{p_1}{p_2}} + \beta_2\xh_i(t)^{\frac{m_1}{m_2}} ) \nonumber \\
	& \leq -\eta ( \beta_1e_i(t)^{\frac{p_1}{p_2}} + \beta_2e_i(t)^{\frac{m_1}{m_2}} ), t\in  [0, \tau_{\max}), \label{eq:f1}
\end{flalign}
where (\ref{eq:f1}) follows that $\exists l \in \cP(i)\cap S$ by Definition \ref{def:longest} and $	\xh_i(t) = x_i(t) - \min_{j \in \cN(i)}\{f_i(x_j(t),w_{ij})\}  \geq  x_i(t) - f_i(x_l(t),w_{il}) =  x_i(t) - f_i(x_l^*,w_{il}) = e_i(t)$. According to Kamke's comparison principle \cite[Theorem \rom{8}]{walter2013ordinary}, $e_i(t)$ is upper bounded by the maximal solution of the following equation:
\begin{equation*}
	\dot{\bar{e}}_i(t) = -\eta ( \beta_1\bar{e}_i(t)^{\frac{p_1}{p_2}} + \beta_2\bar{e}_i(t)^{\frac{m_1}{m_2}} ) , \bar{e}_i(0) = e_i(0).
\end{equation*}
From \cite[Lemma 4.1]{zuo2016distributed}, $e_i(t)$ is upper bounded for all $t \in [0, \tau_{\max})$. Suppose $e_i(t)$ is upper bounded by $B > 0$ for all $i \in \cF_\ell$ with $\ell \in \{1, \cdots, \cD(G) - 2 \}$ over $[0, \tau_{\max})$. For $i \in \cF_{\ell + 1}$ and $j \in \cP(i)$, since $j \in \cF_\ell$ by Definition~\ref{def:longest}, $e_j(t) \leq B$ on $[0, \tau_{\max})$. Hence, by (\ref{eq:singularity}), for all $t \in [0, \tau_{\max})$,
\begin{flalign}
	&\!\!\! \ed_i(t) = \xd_i(t) = -\eta \big( \beta_1\xh_i(t)^{\frac{p_1}{p_2}} + \beta_2\xh_i(t)^{\frac{m_1}{m_2}} )  \big) \nonumber \\
	&\!\!\! \leq  -\eta \big( \beta_1(e_i(t) - LB)^{\frac{p_1}{p_2}} + \beta_2(e_i(t) - LB)^{\frac{m_1}{m_2}} )  \big),  \label{eq:lm1}
\end{flalign}
where (\ref{eq:lm1}) follows that $	\xh_i(t) \geq  x_i(t) - f_i(x_j(t),w_{ij}) \geq x_i(t) - f_i(x^*_j,w_{ij}) - LB = e_i(t) - LB$ with $L$ defined in (\ref{eq:Lipschitz}). By the same argument, $e_i(t)$ is upper bounded over $ [0, \tau_{\max})$, and our claim follows.

{\bf Proof of Lemma \ref{le:globalexistence1}:} Let $[0, \tau_{\max})$ be the right maximal interval of existence for the solutions to (\ref{eq:singularity}). Clarke's generalized derivative \cite{clarke1990optimization} yields
\begin{flalign}
	&	\xd^{(0)}_{\min}(t) = -\eta \sum_{i \in \cK^0(t)}\sigma_i( \beta_1\xh_i(t)^{\frac{p_1}{p_2}} + \beta_2\xh_i(t)^{\frac{m_1}{m_2}} ),  \label{eq:em} 
\end{flalign}
where $\cK^0(t)$ is defined in (\ref{eq:kt}), $\sum_{i \in \cK^0(t)}\sigma_i = 1$ with $\sigma_i \in [0,1]$, and $\xh_i(t) = x_i(t) - f_i(x_l(t),w_{il})$ with $l = \arg \min_{j \in \cN(i)}\{f_i(x_j(t),w_{ij})\}$. Consider two cases: 1) $l \in S$; and 2) $l \notin S$. In case 1), from (\ref{eq:bellman}), 
$\xh_i(t) = x_i(t) - f_i(x_l^*,w_{il}) \leq x_i(t) - x_i^* = e_i(t)$. In case 2), by (\ref{eq:small}), $x_l(t) \geq x_i(t)$. From (\ref{eq:progressive}), $f_i(x_l(t),w_{il}) \geq x_l(t) + \zeta \geq x_i(t) + \zeta$ and $\xh_i(t) \leq -\zeta$. Let $\Delta = \eta(\beta_1\zeta^{\frac{p_1}{p_2}} + \beta_2\zeta^{\frac{m_1}{m_2}}) > 0$ and $\tilde{e} = x^{(0)}_{\min}(t) - \xb_1^\ast$, combining two cases, we have 
\begin{flalign}
	&\xd^{(0)}_{\min}(t) \geq \sum_{i \in \cK^0(t)} \sigma_i \min \{\Delta,  -\eta(\beta_1e_i(t)^{\frac{p_1}{p_2}} + \beta_2e_i(t)^{\frac{m_1}{m_2}}) \} \nonumber \\
	& \geq \sum_{i \in \cK^0(t)} \sigma_i \min \{\Delta,  -\eta(\beta_1\tilde{e}^{\frac{p_1}{p_2}} + \beta_2\tilde{e}^{\frac{m_1}{m_2}}) \} \label{eq:low1} \\
	& = \min \{\Delta, -\eta(\beta_1\tilde{e} ^{\frac{p_1}{p_2}} + \beta_2\tilde{e}^{\frac{m_1}{m_2}}) \}, \label{eq:glt}
\end{flalign}
where (\ref{eq:low1}) follows from $x_i(t) = x^{(0)}_{\min}(t)$ as $i \in \cK^0(t)$, and $x_i^* \geq \xb_1^*$ by Definition \ref{def:S}. 

For (\ref{eq:glt}), consider two cases: 1) $x^{(0)}_{\min}(0) \geq \xb_1^*$; and 2) $x^{(0)}_{\min}(0) < \xb_1^*$. For case 1), as $\Delta > 0$, from Kamke's comparison principle and  \cite[Lemma 4.1]{zuo2016distributed}, $x^{(0)}_{\min}(t)$ is lower bounded by $\xb_1^*$ over $[0, \tau_{\max})$. For the latter case, $x^{(0)}_{\min}(t)$ is strictly increasing before reaching $\xb_1^*$ and remains no less than $\xb_1^*$ thereafter. Therefore, $x^{(0)}_{\min}(t) \geq \min\{\xb_1^*, x^{(0)}_{\min}(0) \}$ for $t \in [0, \tau_{\max})$.

{\bf Proof of Lemma \ref{le:overfinite}:} We first prove that $x_{\min}^{(l)}(t)$ will exceed $\xb_{l+1}^*$ after a finite time $T_f$. For $t \geq t_1$, it follows from (\ref{eq:singularity}) and Clarke's generalized derivative that 
\begin{flalign}
	\xd_{\min}^{(l)}(t) =  -\eta  \sum_{i \in \cK^{l}(t)}\sigma_i( \beta_1\xh_i(t)^{\frac{p_1}{p_2}} + \beta_2\xh_i(t)^{\frac{m_1}{m_2}}  ),\label{eq:clarke1}
\end{flalign}
where $\xh_i(t) = x_i(t) - f_i(x_p(t),w_{ip})$ with $p$ obeying $p = \arg \min_{j \in \cN(i)}\{f_i(x_j(t),w_{ij})\}$. Consider two cases: 1) $p \in \cup_{j \in \{0, 1, \cdots, l \}}\cS_j$; and 2) $p \notin \cup_{j \in \{0, 1, \cdots, l \}}\cS_j$. In case 1),
\begin{flalign}
	&	\xh_i(t) = x_i(t) - f_i(x_p^*,w_{ip})  \leq x_i(t) - x_i^* = e_i(t), \label{eq:gt}
\end{flalign}
where the inequality in (\ref{eq:gt}) uses (\ref{eq:bellman}). In case 2), it follows from (\ref{eq:kt}) that $x_p(t) \geq x_i(t) = x_{\min}^{(l)}(t)$. From (\ref{eq:progressive}), 
\begin{flalign}
	&	\xh_i(t) \leq x_i(t) - f_i(x_p(t),w_{ip})  \leq -\zeta. \label{eq:gz}
\end{flalign}
Let $\Delta = \eta(\beta_1\zeta^{\frac{p_1}{p_2}} + \beta_2\zeta^{\frac{m_1}{m_2}} )  > 0$ and $\breve{e} = x_{\min}^{(l)}(t) -  \xb_{l+1}^*$. Putting (\ref{eq:gt})-(\ref{eq:gz}) into (\ref{eq:clarke1}), from (\ref{eq:glt}) and (\ref{eq:low1}),
\begin{flalign}
	\dot{x}_{\min}^{(l)}(t) \geq   \min\{\Delta, -\eta(\beta_1\breve{e}^{\frac{p_1}{p_2}} + \beta_2\breve{e}^{\frac{m_1}{m_2}} ) \}, \forall t \geq t_1. \nonumber
\end{flalign}
If $x^{(l)}_{\min}(t_1) \geq \xb_{l+1}^* - \zeta$, by \cite[Lemma 4.1]{zuo2016distributed},  $x_{\min}^{(l)}(t) \geq \xb_{l+1}^*$ for $t \geq t_1 + T_*$. If $x^{(l)}_{\min}(t_1) < \xb_{l+1}^* - \zeta$, $x^{(l)}_{\min}(t)$ first increases at least at a linear rate $\Delta$ until reaching $\xb_{l+1}^* - \zeta$, and then exceeds $\xb_{l+1}^*$ after an additional $T_*$, i.e., $x_{\min}^{(l)}(t) \geq \xb_{l+1}^*$ for $t \geq t_1 + T_* + T_f$ with
\begin{flalign*}
	T_f \leq \frac{\xb_{l+1}^* - \zeta - x^{(l)}_{\min}(t_1)}{\Delta} \leq \frac{\xb_R^* - \zeta - \min\{\xb_1^*, x^{(0)}_{\min}(0)\}}{\Delta}
\end{flalign*}
where the last inequality uses Lemma \ref{le:globalbound} and Definition \ref{def:S}. Therefore, $x_{\min}^{(l)}(t) \geq \xb_{l+1}^*$ for $t \geq t_1 + T_* + \max\{0, T_f\}$.

We next prove that $x_i(t)$, for $i \in \cS_{l + 1}$, will converge to $x_i^*$ with another $T_*$. Let $p = \arg \min_{j \in \cN(i)}\{f_i(x_j(t),w_{ij})\}$. Then for $t \geq t_1 + T_* + \max\{0, T_f\}$,
\begin{flalign}
	f_i(x_p(t),w_{ip})  \geq  x_i^*,\label{eq:twop}
\end{flalign}
where (\ref{eq:twop}) follows as either $p \in \cup_{j = 0}^{l} \cS_j$, in which case $f_i(x_p(t),w_{ip}) = f_i(x_p^*,w_{ip}) \geq x_i^*$ by (\ref{eq:bellman}), or $p \in V\setminus \cup_{j = 0}^{l} \cS_j$, in which case $f_i(x_p(t),w_{ip}) \geq f_i(\xb_{l+1}^*,w_{ip}) > \xb_{l+1}^* + \zeta > x_i^*$ by the properties of $f_i$. Note that any $j \in \cP(i)$ is in $\cS_q$ with $q < l + 1$, and by assumption $x_j(t) = x_j^*$ for $t \geq t_1$, implying that $f_i(x_j(t),w_{ij}) = f_i(x_j^*,w_{ij}) = x_i^*$ for $t \geq t_1$. Thus, by (\ref{eq:twop}), $\xh_i(t) =  x_i(t) - x_i^* = e_i(t)$ for $t \geq t_1 + T_* + \max\{0, T_f\}$, and by (\ref{eq:singularity}) $\xd_i(t) = \ed_i(t) = -\eta ( \beta_1e_i(t)^{\frac{p_1}{p_2}} + \beta_2e_i(t)^{\frac{m_1}{m_2}} )$.   
Thus, $x_i(t) = x_i^*$ for $t \geq t_1 + 2T_* + \max\{0, T_f\}$ by \cite[Lemma 4.1]{zuo2016distributed}.

{\bf Proof of Theorem \ref{the:fix}:} 
As for all $i \in \cS_0 = S$, $x_i(t) = x^*_i$ for all $t \geq 0$, by (\ref{le:overfinite}) for all $i \in \cS_1$, $x_i(t) = x_i^*$ for $t \geq 2T_* + \max\{0, T_f\}$. Repeating the above procedure, $x_i(t) = x_i^*$ for all $i \in \cup_{j = 0}^l \cS_j$ and $t \geq 2lT_* + l\max\{0, T_f\}$ with $l \in \{1, \cdots, R \}$. As $R \leq |V\setminus S|$, our claim follows.

{\bf Proof of Theorem \ref{the:singularity}:} Define $\tau$ as in (\ref{eq:ts}) and let $y_i(\tau) = x_i(t(\tau))$ with $x_i(t)$ defined in (\ref{eq:fully}). Consider $\xh_i(t) \neq 0$ with $i \notin S$. From (\ref{eq:fully}) and (\ref{eq:range}), for all $\tau \in [0, +\infty)$,
\begin{flalign}
	&	\frac{dy_i(\tau)}{d\tau} = \frac{dx_i(t)}{dt}\frac{dt}{d\tau}
 = -(\beta_1(\yh_i(\tau))^{\frac{p_1}{p_2}} + \beta_2(\yh_i(\tau))^{\frac{m_1}{m_2}} ), \nonumber
\end{flalign}
where  $\yh_i(\tau) = y_i(\tau) - \min_{j \in \cN(i)}\{f_i(y_j(\tau),w_{ij})\}.$ Therefore, it follows from Theorem \ref{the:fix} that for all $i \in V$
\begin{equation*}
	y_i(\tau) = x_i^*, \forall \tau \geq 2|V\setminus S|T'_* + |V\setminus S|\max\{0, T_f\},
\end{equation*}
with $T'_* = \frac{1}{\beta_1}\frac{m_2}{m_1 - m_2} + \frac{1}{\beta_2}\frac{p_1}{p_2 - p_1}$ as established in Lemma \ref{le:overfinite}, $T_f$ defined in (\ref{eq:Tf}), 
and $x_i^*$ defined in (\ref{eq:bellman}).

Let $T_p^*$ satisfy $\tau(T_p^*) = 2|V\setminus S|T'_* + |V\setminus S|\max\{0, T_f\}$. Obviously, $T_p^* < T_p$ and $x_i(T_p^*) = x_i^*$ for all $i \in V$. Moreover, by (\ref{eq:bellman}) and (\ref{eq:xhat}),  $\xh_i(T_p^*) = 0$ for $i \notin S$. Thus, the third case in (\ref{eq:fully}), namely the fixed time control strategy, is activated, with $x_i(T_p^*) = x_i^*$ for all $i \in V$, implying that $x_i(t) = x_i^*$ for all $i \in V$ and all $t \geq T_p^*$.

\end{document}